\documentclass[technote,9pt,letterpaper,final]{IEEEtran}
\usepackage{amsmath,amssymb,dsfont}
\usepackage{graphicx}

 \IEEEoverridecommandlockouts
 \IEEEoverridecommandlockouts

\newtheorem{theorem}{Theorem}

\newtheorem{example}[theorem]{Example}

\newtheorem{lemma}[theorem]{Lemma}

\newtheorem{remark}[theorem]{Remark}

\def\Cbb{\mathbb{C}}

\def\Hcal{\mathcal{H}}

\def\diag{{\rm diag}}

\def\Acal{\mathcal{A}}

\def\Vcal{\mathcal{V}}

\def\Rbb{\mathbb{R}}
\def\Cbb{\mathbb{C}}

\def\diag{{\rm diag}}
\def\tr{{\rm tr}}
\def\isym{\mathfrak{i}}

\title{On the infeasibility of entanglement generation in Gaussian quantum systems via classical control\footnote{ Supported by the
Australian Research Council and Air Force Office of Scientific Research (AFOSR). This material is based on research sponsored by the Air Force Research Laboratory, under agreement number FA2386-09-1-4089.  The U.S. Government is authorized to reproduce and distribute reprints for Governmental purposes notwithstanding any copyright notation thereon.
The views and conclusions contained herein are those of the authors and should not be interpreted as necessarily representing the official policies or endorsements, either expressed or implied, of the Air Force Research Laboratory or the U.S. Government.
}}

\author{Hendra~I.~Nurdin, Ian~R.~Petersen, and Matthew~R.~James
\thanks{H. I. Nurdin is with the Research School of Engineering, The Australian National University, Canberra, ACT 0200,
Australia. Email: Hendra.Nurdin@anu.edu.au. Research supported by the Australian Research Council.}
\thanks{ I. R. Petersen is with the School of Information Technology and Electrical Engineering, University of New South Wales at the Australian Defence Force Academy, Canberra, ACT 2600, Australia. Email: i.r.petersen@gmail.com.}
\thanks{M.~R.~James is with the ARC Centre for Quantum Computation and Communication Technology, Research  School of Engineering, The Australian
    National University, Canberra, ACT 0200, Australia. Email: Matthew. James@anu.edu.au. Research supported by the
    Australian Research Council and AFOSR Grant FA2386-09-1-4089 AOARD 094089.}
}    

\date{}

\begin{document}
\maketitle

\begin{abstract}
 This paper uses a system theoretic approach to show that classical linear time invariant controllers cannot generate steady state entanglement in a bipartite Gaussian quantum system which is initialized in a Gaussian state.  The paper also shows that the use of classical linear controllers  cannot generate entanglement in a finite time from a bipartite system initialized in a separable Gaussian state. The approach reveals  connections between system theoretic concepts and the well known physical principle that local operations and classical communications  cannot generate entangled states starting from separable states.
\end{abstract}

\IEEEpeerreviewmaketitle

\section{Introduction}
\label{sec:intro}
 Entanglement is a unique feature of quantum mechanical systems not found in classical systems and is responsible for some of their predicted counterintuitive behavior, as exemplified by the famous Einstein-Podolsky-Rosen paradox \cite{NC00}. Entanglement gives rise to experimentally verifiable non-classical correlations among measurement statistics \cite{AGR82} that cannot be explained by the usual classical probability models. One well known application of entanglement
is {\em quantum teleportation}, the process of transferring the {\em unknown} state of one quantum system to another whilst destroying the state of the former, without the two quantum systems ever interacting directly with one another \cite{BBCJPW93}. This process is at the heart of quantum communication schemes. 

A bipartite quantum system  is the composite of two quantum systems. The state of such a system will be referred to as a bipartite state and is  represented by a density operator $\rho$\footnote{$\rho$ is a self-adjoint, positive semidefinite operator on an underlying Hilbert space with $\tr(\rho)=1$; e.g., see \cite{NC00}. }. Suppose that the system is composed of a quantum system A with underlying Hilbert space $\Hcal_A$ and a quantum system B with underlying Hilbert space $\Hcal_B$. A state $\rho$ is said to be {\em separable} if it  can be decomposed as $\rho =\sum_{k} p_k \rho^{A}_k \otimes \rho^{B}_k$, with $\rho_k^A$ and $\rho_k^B$ being density operators on $\Hcal_A$ and $\Hcal_B$, respectively, for $k=1,2,\ldots$. Here $\otimes$ denotes the tensor product of operators. If a bipartite state is not separable, then it is said to be {\em entangled}. 
For a pure state density operator $\rho$ (i.e., $\tr(\rho^2)=1$),  it can  be easily 
determined  if it is separable; e.g., see \cite{PV07}. However, determining the separability of a mixed state bipartite density operator
$\rho$ (i.e., $\tr(\rho^2)<1$),  is far from straightforward and a complete characterization is only known for certain types of bipartite systems, such as for bipartite systems on
the finite dimensional Hilbert space $\Cbb^2 \otimes \Cbb^2$. In fact, the general problem of determining the separability of a given mixed quantum state is known to be
NP-hard \cite{Gurv04}. Another class of bipartite systems
for which a complete characterization of separability is known is the
class of bipartite Gaussian systems \cite{Sim00,Hole75,Ades07,PSL09}. These systems are commonly
encountered in the field of quantum optics. For such systems, the underlying Hilbert space is the tensor product of two quantum harmonic oscillator Hilbert spaces; e.g., see \cite[Chapter III]{Mey95}. Also, the separability of a state can be completely determined from  the (symmetrized) covariance matrix
of the canonical position and momentum operators of the system
\cite{Sim00}.  The class of systems considered in this paper
is the class of bipartite Gaussian systems.  In particular, we  analyze {\em dynamical} bipartite   Gaussian
quantum systems whose covariance matrices evolve in time.  Hence, the
separability or entanglement of these systems also evolves in time. In
quantum optics, these   dynamical bipartite Gaussian systems correspond
to a class of linear quantum stochastic systems
\cite{JNP06,NJP07b,NJD08} that  are driven by Gaussian bosonic fields
and with a density operator initially in a Gaussian
state. 
The dynamics of such systems can be  represented  by linear quantum
stochastic differential equations  (QSDEs) in the canonical position and momentum operators
and this makes them suitable for  a system-theoretic
analysis. 

We study the problem of entanglement generation using classical finite dimensional  (linear time-invariant (LTI) and time varying) 
controllers from a system-theoretic point of view. The main contribution of the paper is 
the use of system theoretic arguments and methods to show that the
application of a classical linear dynamic controller cannot induce
entanglement in a  dynamical bipartite Gaussian  system
which is initially in a separable state. 
Our result is in agreement with the fundamental
physical principle that  Local Operations and Classical Communication
(LOCC) cannot generate entanglement between initially separable
states; e.g., see \cite{PV07} for a proof of this  result. 
One motivation for the results of this paper is that they provide a starting point for
investigating  connections between systems theory and quantum
physical principles. The no-go results for Gaussian quantum systems
considered here are in a similar spirit to other no go results that have
previously been obtained in \cite{SSHKST87}, showing that linear modulation of a
beam cannot create out-of-loop squeezing, and \cite{WM94}, showing that neither in-cavity squeezing nor output squeezing  can be created
using linear modulation of the cavity field. 

\section{Preliminaries}
\label{sec:prelim} 

\subsection{Notation}
We will use the following notation: $\isym=\sqrt{-1}$, $^*$ denotes the adjoint of a linear operator
as well as the conjugate of a complex number. If $A=[a_{jk}]$ then $A^{\#}=[a_{jk}^*]$, and $A^{\dag}=(A^{\#})^T$, where $^T$ denotes matrix transposition.  $\Re\{A\}=(A+A^{\#})/2$ and $\Im\{A\}=\frac{1}{2\isym}(A-A^{\#})$.
We denote the identity matrix by $I$ whenever its size can be
inferred from context and use $I_{n}$ to denote an $n \times n$
identity matrix. Similarly, $0$ denotes  a matrix with zero
entries whose dimensions can be determined from context. 
$\diag(M_1,M_2,\ldots,M_n)$  denotes a block diagonal matrix with
square matrices $M_1,M_2,\ldots,M_n$ on its diagonal, and $\diag_{n}(M)$
denotes a block diagonal matrix with the
square matrix $M$ appearing on its diagonal blocks $n$ times. Also, we
will let $J=\left[\begin{array}{rr}0 & 1\\-1&0\end{array}\right]$.

\subsection{The class of linear quantum stochastic systems}
\label{sec:linear-summary} 
In this paper, we are concerned with a class of quantum stochastic
models of open (i.e., quantum systems  that can interact with an
environment) {\em Markov} quantum systems  that are widely used  in
the area of  quantum optics. Such models have been used in the physics
and mathematical physics literature since at least the 1980's; e.g., see \cite{HP84,GC85,KRP92,GZ00,WM10}. We focus on the special sub-class of {\em linear} quantum stochastic models (e.g.,  see \cite[Section 7.2]{WaM94}, \cite[Section 6.6]{WM10}, \cite[Sections 3, 3.4.3, 5.3, Chapters 7 and 10]{GZ00}, \cite{EB05},  \cite[Section 5]{BE08}, \cite{JNP06,NJP07b,NJD08,Mab08}) that describe the Heisenberg evolution of the (canonical) position and momentum operators of several independent open quantum harmonic oscillators that are coupled to external coherent bosonic fields, such as coherent laser beams. 
These linear stochastic models describe quantum optical devices such as optical cavities \cite[Section 5.3.6]{BR04}\cite[Chapter 7]{WaM94}, linear quantum amplifiers \cite[Chapter 7]{GZ00}, and finite bandwidth squeezers \cite[Chapter 10]{GZ00}. Following  \cite{JNP06,NJP07b,NJD08}, we will refer to this class of models as {\em linear quantum stochastic systems}.

Suppose we have $n$ independent quantum harmonic oscillators. The $j$th quantum harmonic oscillator has position and momentum operators
$q_j$ and $p_j$ with underlying Hilbert space $L^2(\mathbb{R})$; see, e.g., \cite[Chapter III]{Mey95}. The position and momentum operators satisfy the canonical commutation relations $[q_j,p_k]=2\isym \delta_{jk}$, $[q_j,q_k]=0$, and $[p_j,p_k]=0$, where $\delta_{jk}$ denotes the Kronecker delta and $[\cdot,\cdot]$ denotes the commutation operator. 
The  quantum harmonic oscillators are assumed to be
coupled to $m$ external independent quantum fields modelled by bosonic annihilation field operators $\mathcal{A}_1(t), \mathcal{A}_2(t),\ldots,\mathcal{A}_m(t)$ which are defined on  separate Fock spaces (over $L^2(\Rbb)$) for each field operator  \cite{HP84,KRP92}. For each annihilation field operator $\mathcal{A}_j(t)$, there is a corresponding creation field operator $\mathcal{A}_j^*(t)$, which is defined on the same Fock space and is the operator adjoint of $\mathcal{A}_j(t)$.
The field operators are adapted quantum stochastic processes with forward differentials $d\mathcal{A}_j(t)=\mathcal{A}_j(t+dt)-\mathcal{A}_j(t)$ and $d\mathcal{A}_j^*(t)=\mathcal{A}_{j}^*(t+dt)-\mathcal{A}_j^*(t)$ that have the quantum It\^{o} products \cite{HP84,KRP92}:
$d\mathcal{A}_{j}(t)d\mathcal{A}_{k}(t)^*=\delta_{jk}dt;\, 
d\mathcal{A}_{j}^*(t)d\mathcal{A}_{k}(t) = d\mathcal{A}_{j}(t)d\mathcal{A}_{k}(t)=d\mathcal{A}_{j}^*(t) d\mathcal{A}_{k}^*(t)=0. $
 
We collect the position and momentum operators in the column vector
$x$ defined by $x=(q_1,p_1,q_2,p_2,\ldots,q_n,p_n)^T$. Note that 
we may write the canonical commutation relations as
$xx^T-(xx^T)^T=2\isym\Theta$ with $\Theta=\diag_{n}(J)$. We take the
composite system of $n$ quantum harmonic oscillators to have a {\em
  quadratic Hamiltonian} $H$ given by $H=\frac{1}{2} x^T R x$, where
$R$ is a real $2n \times 2n$ symmetric matrix. The quantum harmonic oscillators are
coupled to the $k$-th quantum field  via the  singular interaction
Hamiltonian $H_k=\isym (L_k \eta_k^*(t)-L_k^*\eta_k(t))$
\cite{GC85,GZ00}, where $L_k=K_k x$ (with $K_k \in \Cbb^{1 \times 2n}$)
is a linear coupling operator describing the linear coupling of the quantum harmonic oscillator position and momentum
operators to $\eta_k(t)$. Here $\eta_k(t)$ is a {\em quantum white noise process} \cite{GC85,GZ00} satisfying the relation $\mathcal{A}_k(t)=\int_{0}^{t} \eta_k(s)ds$. We now collect the coupling operators $L_1,L_2,\ldots,L_m$ together in one {\em linear coupling vector} $L= (L_1,L_2,\ldots,L_m)^T=K x$, with $K=[\begin{array}{cccc} K_1^T & K_2^T & \ldots & K_m^T\end{array}]^T$, and collect the field operators together as $\mathcal{A}(t)=(\mathcal{A}_1(t),\mathcal{A}_2(t),\ldots,\mathcal{A}_m(t))^T$. Then 
the {\em joint} evolution of the oscillators and the quantum fields is given by a unitary adapted process $U(t)$ satisfying the Hudson-Parthasarathy QSDE \cite{HP84,KRP92}:
\begin{eqnarray*}
dU(t) &=& \biggl({\rm tr}\bigl((S-I)^T d\Lambda(t)\bigr) +  d\mathcal{A}(t)^{\dag} L - L^{\dag}Sd\mathcal{A}(t)\\
&&\quad -(\isym H + \frac{1}{2}L^{\dag}L )dt \biggr)U(t),
\end{eqnarray*}
where $S \in \Cbb^{m \times m}$ is a complex unitary matrix (i.e., $S^{\dag}S=SS^{\dag}=I$) called the {\em scattering matrix}, and $\Lambda(t)=[\Lambda_{jk}(t)]_{j,k=1,\ldots,m}$, with $\Lambda_{kj}(t)=\Lambda_{jk}(t)^*$. The processes $\Lambda_{jk}(t)$ for $j,k=1,\ldots,m$ are adapted quantum stochastic processes referred to as {\em gauge processes}, and the forward differentials $d\Lambda_{jk}(t)=\Lambda_{jk}(t+dt)-\Lambda_{jk}(t) $ $j,k=1,\ldots,m$ have the quantum It\^{o} products:
\begin{eqnarray*}
d\Lambda_{jk}(t)
d\Lambda_{j'k'}(t)\hspace*{-1pt}&=&\hspace*{-1pt}\delta_{kj'}d\Lambda_{jk'}(t), \ 
d\mathcal{A}_j(t)
d\Lambda_{kl}(t)\hspace*{-1pt}=\hspace*{-1pt}\delta_{jk}d\mathcal{A}_l(t),\\
d\Lambda_{jk} d\mathcal{A}_l(t)
^*\hspace*{-1pt}&=&\hspace*{-1pt}\delta_{kl}d\mathcal{A}_j^*(t),
\end{eqnarray*}
with all other remaining cross products between $\mathcal{A}_j(t),\mathcal{A}^*_k(t),\Lambda_{j'k'}(t)$ being $0$.

For any adapted processes $X(t)$ and $Y(t)$ satisfying a quantum It\^{o}
stochastic differential equation, we have the {\em quantum It\^{o} rule}
$d(X(t)Y(t))=X(t)dY(t)+(dX(t))Y(t) + dX(t) dY(t)$; e.g., see
\cite{HP84,KRP92}.  Using the quantum It\^{o} rule and the quantum It\^{o} products
given above, as well as exploiting the canonical commutation relations
between the operators in $x$, the {\em Heisenberg evolution}
$x(t)=U(t)^* x U(t)$ of the canonical operators in the vector $x$ 
can be obtained. Then $x(t)$  satisfies the QSDE (see \cite{EB05,BE08,JNP06,NJD08}):
\begin{align}
dx(t)&= d(U(t)^* x U(t)), \notag \\
&=  A_ox(t)dt+ B_o\left[\small \begin{array}{c} d\mathcal{A}(t)
\\ d\mathcal{A}(t)^{\#} \end{array} \normalsize \right];  
x(0)=x, \notag\\
dY(t)&= d(U(t)^* \mathcal{A}(t)U(t)), \notag \\
&=  C_o x(t)dt+  D_o d\mathcal{A}(t), \label{eq:qsde-out}
\end{align}
with $A_o=2\Theta(R+\Im\{K^{\dag}K\})$, $B_o=2\isym \Theta [\begin{array}{cc}
-K^{\dag}S & K^TS^{\#}\end{array}]$,
$C_o=K$, and $D_o=S$. Here,  $Y(t)=(Y_1(t),\ldots,Y_m(t))^T=U(t)^* \mathcal{A}(t) U(t)$ is a vector of
{\em output fields} that results from the interaction of the quantum
harmonic oscillators and the incoming quantum fields
$\mathcal{A}(t)$. Note that the dynamics of $x(t)$ is linear, and
$Y(t)$ depends linearly on $x(t)$. We refer to $n$ as the number of {\em
  degrees of freedom} of the linear quantum stochastic system.

In this paper it will be  convenient to write the dynamics in quadrature form as in  \cite{JNP06}:
\begin{align}
dx(t)&=Ax(t)dt+Bdw(t);\, x(0)=x. \nonumber\\
dy(t)&= C x(t)dt+ D dw(t), \label{eq:qsde-out-quad}
\end{align}
with
\begin{eqnarray*}
w(t)&=&2(\Re\{\Acal_1(t)\},\Im\{\Acal_1(t)\},\ldots,\Re\{\Acal_m(t)\},\Im\{\Acal_m(t)\})^T; \\
y(t)&=&2(\Re\{Y_1(t)\},\Im\{Y_1(t)\},\ldots,\Re\{Y_m(t)\},\Im\{Y_m(t)\})^T.
\end{eqnarray*}
Here, the real matrices $A,B,C,D$ are in a one to one correspondence
with the matrices 
$A_o,B_o,C_o,D_o$. Also, the quantity $w(t)$
satisfies the It\^{o} relationship $dw(t)dw(t)^T = F_wdt$ where $F_w \geq
0$; see \cite{JNP06}. Furthermore, we define the matrix $S_w =
\frac{1}{2}\left(F_w + F_w^\#\right)$ and the differential commutation
matrix  $T_w = \frac{1}{2}\left(F_w -F_w^\#\right)$. For the boson fields that we consider here, $T_w$ is necessarily of the form $T_w=\isym \diag_{m}(J)$. 
The symmetric matrix $S_w$ is then such that $F_w \geq 0$ and for the Gaussian boson fields that are of interest here this matrix reflects the statistics
of the field. For instance, $S_w=I$ corresponds to a vacuum Gaussian boson field.

Fig.~\ref{fig:two-cavities} shows an example of a two degree of freedom linear quantum stochastic system connected to a classical controller.   The linear quantum stochastic system consists of two independent  optical cavities \cite{BR04,WaM94,NJD08} denoted by $G_1$ and $G_2$. The two optical cavities are connected to the classical controller via a homodyne detector (HD)  which measures one of the quadratures of the output field $Y_{1}(t)$ from $G_1$, and an electro-optic modulator (MOD) which modulates the quantum field $\mathcal{A}_{3}(t)$ with the controller output signal $u(t)$ and then sends the resulting field $\mathcal{A}_{2}(t)$  to $G_2$. 
\begin{figure*}[htbp]
\centering
\includegraphics[scale=0.30]{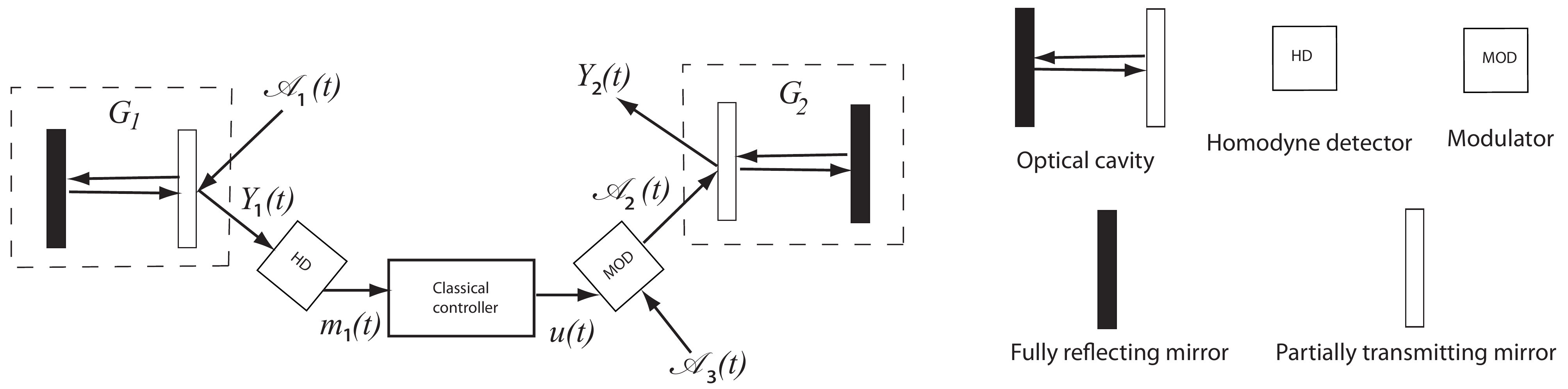}
\caption{The interconnection of two optical cavities $G_1$ and $G_2$ via  a classical
  controller. } \label{fig:two-cavities}
\end{figure*}

\begin{remark}
\label{rm:two-df-terms}  For the remainder of this paper, we will consider  the case where $n=2$, corresponding to two degree of freedom linear quantum stochastic 
systems, with the quantum harmonic oscillators being initialized in a Gaussian state\footnote{A state with density operator $\rho$, is said to be {\em Gaussian} if $\tr(\rho
e^{\isym \lambda^Tx})=e^{\isym \lambda^T m-\frac{1}{2}\lambda^T G
\lambda}$ for all $\lambda \in \Rbb^{2n}$, where $m \in \Rbb^{2n}$ and
$G$ is a real symmetric matrix satisfying $G+\isym \Theta \geq 0$ with $\Theta$ as
defined previously; see, e.g., \cite{Lin98,GZ00,KRP92}.}. 
Also, for such a  linear quantum stochastic system $G$,  we define a  {\em (linear) dynamical bipartite  Gaussian quantum system (corresponding to  $G$)} as the  open quantum system obtained from $G$ by tracing out (averaging) the bosonic fields.
 \end{remark}

\section{Simple system-theoretic proof of the Heisenberg Uncertainty Principle}
\label{sec:proof-HUP}
For a  system of the form
(\ref{eq:qsde-out-quad}), the corresponding symmetric covariance matrix 
defined by 
\[
P(t)=\frac{1}{2}\tr(\rho(0)(x(t)x(t)^T+(x(t)x(t)^T)^T))
\]
varies with time.
Here, $ \rho(0)$ is the initial density operator of the
overall composite closed system.
In this section, we will assume that the matrix $A$ in (\ref{eq:qsde-out-quad}) is Hurwitz.  Then the  steady-state symmetrized covariance matrix $P=\mathop{\lim}_{t \rightarrow \infty} P(t)$ satisfies the {\em real} Lyapunov
equation (see, e.g., \cite[p. 327]{WM10}, \cite[Section 4]{NJP07b}): 
\begin{equation}
A P + P A^T + BS_wB^T =0. \label{eq:cov-ARE}
\end{equation}
On the other hand, since the commutation relations are preserved, we also
have that \cite{JNP06}:
\begin{equation}
A \Theta + \Theta A^T -\isym B T_w B^T =0. \label{eq:comm-cond}
\end{equation}

Defining the complex Hermitian matrix  $\tilde P=P+\isym \Theta$, we see from combining  (\ref{eq:cov-ARE}) and (\ref{eq:comm-cond}) that $\tilde P$ satisfies the {\em complex} Lyapunov equation:
$A\tilde P + \tilde P A^T + BF_wB^T=0$,  
where $F_w=S_w +T_w$. Since $F_w \geq 0$  and $A$ is Hurwitz, it
follows that $\tilde P \geq 0$; e.g., see \cite{Bern05}. Equivalently, in terms of $P$ and $\Theta$ we have that:
 $P+\isym \Theta \geq 0$.  
This matrix inequality is a version of the  Heisenberg
 Uncertainty Principle that must be satisfied by all Gaussian quantum
 systems; e.g., see \cite{Sim00,Hole75} for this alternate form of the Heisenberg Uncertainty Principle.  

\section{Classical LTI controllers cannot generate steady-state bipartite entanglement in linear Gaussian quantum systems}

\label{sec:main-results}
\subsection{Separability criterion for dynamical bipartite  Gaussian systems}
\label{sec:sep-criterion}

It has been shown in \cite{Sim00} that the separability 
of a bipartite Gaussian density operator $\rho$ can be completely determined from a
complex linear matrix inequality (LMI) involving the (symmetrized) covariance
matrix $P=\frac{1}{2}\tr(\rho (xx^T+(xx^T)^T))$; see also \cite{Ades07,PSL09}. 

\begin{lemma}[\cite{Sim00,Ades07,PSL09}]
\label{lem:sep-covar} A bipartite Gaussian density operator $\rho$ is
separable if and only if the corresponding covariance matrix $P$ satisfies the LMI 
$P+\isym \diag(J,-J) \geq 0$.
\end{lemma}

Note here that  without loss of generality, we
can assume that $x$ has zero mean because the mean of $x$
plays no role in determining the separability of the associated density
operator. 
Now, in the case of a dynamical bipartite  Gaussian quantum system corresponding to a linear quantum stochastic system, the covariance
matrix can vary with time and is given by 
\begin{eqnarray}
\label{cov_mat}
P(t)&=& 
\frac{1}{2}\tr(\rho_{\rm o}(t)
(xx^T+(xx^T)^T))\nonumber \\
&=&\frac{1}{2}\tr(\rho(t)(xx^T+(xx^T)^T))\nonumber \\
&=&\frac{1}{2}\tr(\rho(0)(x(t)x(t)^T+(x(t)x(t)^T)^T)),
\end{eqnarray}
where $ \rho(t)$ is the density operator  at time $t \geq 0$ of the
overall composite closed system, while $\rho_{\rm o}(t)$ is  the
reduced density operator of the two quantum harmonic oscillators at
time $t$ obtained 
by tracing out the bosonic fields; e.g., see
\cite{GZ00}.  Note that the
second equality in (\ref{cov_mat}) follows from the definition of the partial trace
(e.g., see \cite[p. 102]{KRP92}) 
since the elements of $x$  are operators on the bipartite quantum
harmonic oscillator Hilbert space. Also, the final equality in (\ref{cov_mat})
follows by  switching from the Schr\"{o}dinger picture (in which $\rho(t)$
evolves in time) to the Heisenberg picture (in which $x(t)$ evolves in time
and the overall density operator is fixed as $\rho(0)$). Thus, to check
whether the system is separable at any time $t \geq 0$, it is
equivalent to 
check if the LMI $P(t) + \isym \diag(J,-J) \geq 0$ is satisfied at
that time $t$. 

\subsection{Separability of dynamical bipartite Gaussian systems coupled via a classical LTI controller}
Let $G_1=(A_1,B_1,C_1,D_1)$ and $G_2=(A_2,B_2,C_2,D_2)$ define two linear quantum
stochastic systems of the form (\ref{eq:qsde-out-quad}). We form a linear quantum stochastic system $G=(A,B,C,D)$ of the form (\ref{eq:qsde-out-quad}) from $G_1$ and $G_2$, with
$A=\diag(A_1,A_2)$, $B=\diag(B_1,B_2)$, $C=\diag(C_1,C_2)$, and $D=\diag(D_1,D_2)$. We then obtain a  dynamical bipartite Gaussian quantum system corresponding to $G$ (see Remark \ref{rm:two-df-terms}). 
The quantum system $G$ is connected to a finite dimensional classical controller as shown in Fig. \ref{fig:lqt} to form a quantum feedback control system. In this quantum feedback control system, some of the output fields from $G_1$ and $G_2$ are measured and fed to the classical controller 
that processes these
measurements linearly to produce  control signals that  are fed back
into $G_1$ and/or $G_2$.  Here, control actuation can be facilitated in
two ways: 
\begin{enumerate}
\item Modulating the Hamiltonian of $G_k$ by a classical $2 \times 1$ signal vector $u_{1,k}$. If the canonical operators of $G_k$ are represented by a vector of operators $x_k=(q_k,p_k)^T$, this means that the quadratic  Hamiltonian $H_k$ of $G_k$ is augmented by adding a linear (time varying) Hamiltonian term $H_{l,k}(t)$ of the form $H_{l,k}(t)= u_{1,k}(t)^T M_k x_k$, where $M_k$ is a real $2 \times 2$ matrix. Thus, the total Hamiltonian for $G_k$ becomes  $H_k+H_{l,k}(t)$. The signal $u_{1,k}(t)$ is classical and can depend linearly on the classical controller internal  variables (i.e., its state) as well as the measurement results. This actuation can be implemented in different ways, for instance, as described in the Appendix of \cite{YNJP08}. 

\item Modulating (or displacing) an input field of $G_k$ with a classical control signal $u_{2,k}(t)$. This can be implemented by an electro-optic modulator.

\end{enumerate}
In the quantum feedback control system
shown in Fig.~\ref{fig:lqt},  the vector of quantum input fields $\Acal_k(t)$ for the system
$G_k$ is partitioned into two parts: some of which will be the
components of $\Acal_{k1}(t)$ while the others will be components of
$\Acal_{k2}(t)$. Here $\Acal_{k1}(t)$ represents the input fields of  $G_k$
that will not be modulated by the controller,  while $\Acal_{k2}(t)$
represents  the input fields that are modulated by the controller.
Part of the output  vector of quantum signals, $Y_{k2}(t)$, of $G_k$ ($k=1,2$)
is passed through a network of static optical components (as listed
in \cite[Section 6.2]{NJD08}) and homodyne detectors (labelled as
HDN in the diagram) that produces the  set of classical measurements
signals $m_k(t)$ which drive the controller. 
The controller produces
two sets of classical control signals ($k=1,2$): one set,
$u_{k1}(t)$,  modulates the linear Hamiltonian term $H_{1,k}$, and
another set, $u_{k2}(t)$, is modulated by a network of (possibly
electro-optic) modulators (denoted in the diagram by MOD)  to produce
the quantum signal  $\Acal_{k2}(t)$ as one of the input fields into
$G_k$. The signals $\Vcal_{jk}(t)$, $j,k=1,2$ are any additional
quantum noises required for the operation of HDN and MOD (they may be
suitably absorbed into the definition of $w_1$ or $w_2$).  
\begin{figure}[h!]
\centering
\includegraphics[scale=0.35]{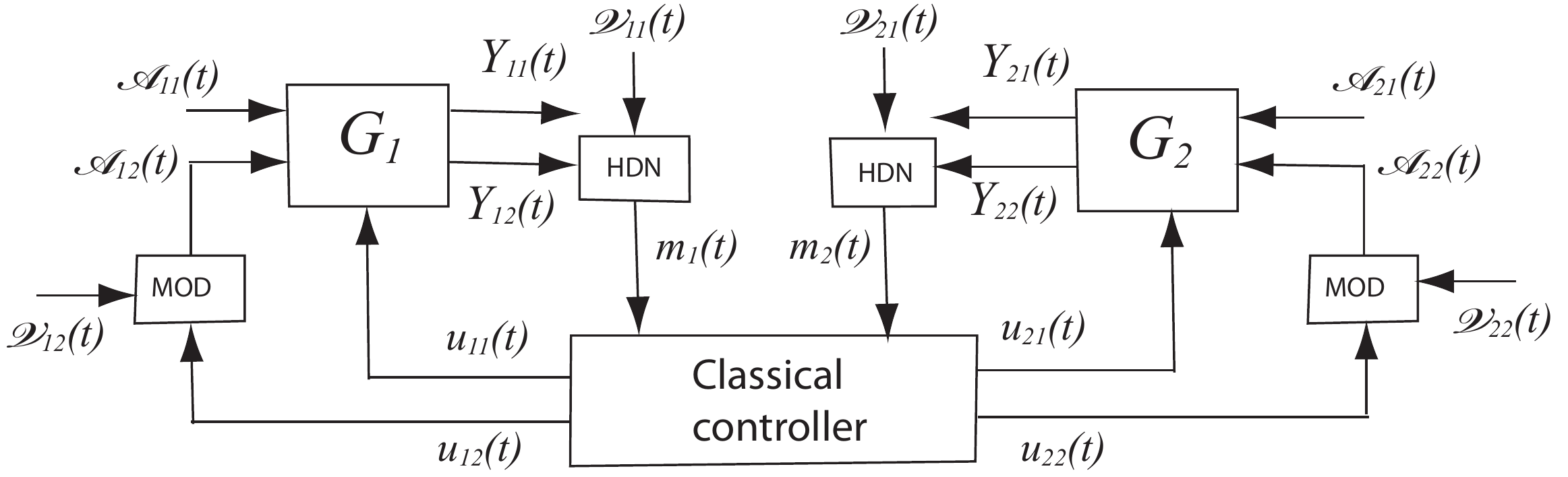}
\caption{Interconnection of $G_1$ and $G_2$ via  a classical
  controller.} \label{fig:lqt}
\end{figure}

The  assumptions that we will use regarding this quantum feedback control system are:
\begin{enumerate}
\item The control  $u(t)=(u_{1,1}(t),u_{2,1}(t),u_{1,2}(t),u_{2,2}(t))^T$ has been generated by a finite dimensional linear (time invariant or time varying)  system. 

\item The quantum signals coming into $G_1$ and $G_2$ come from
  independent sources. Therefore,
  $F_w=\diag(F_{w_1},F_{w_2})$ and hence,
  $S_w=\diag(S_{w_1},S_{w_2})$, $T_w=\diag(T_{w_1},T_{w_2})$. 

\end{enumerate}

Note that the systems $G_1$ and $G_2$ are not directly connected to
one another.  That is, no output field from $G_1$ is passed directly
to $G_2$ and vice-versa. They are only indirectly connected via the
classical controller. Note also that  the overall
closed-loop system is then  a mixed classical-quantum linear stochastic system as described in \cite{JNP06}. 

 Let $z(t)$ denote the
controller internal state which is classical in nature and  of arbitrary dimension $n_c$.  We could also allow the classical controller to be driven by
an additional classical Wiener noise source $w_c(t)$ that is not
derived from the measurement signals. However, this additional noise
may be absorbed into $w_1$ or $w_2$; see \cite{JNP06} for details. Now,
let $x(t)=(x_1(t),x_2(t),z(t))^T$ and $w(t)=(w_1(t),w_2(t))^T$ where $x_1(t)$ represents the vector of system variables for the quantum system $G_1$ and $w_1(t)$ represents the vector of quantum noise inputs for $G_1$. Also, $x_2(t)$ represents the vector of system variables for the quantum system $G_2$ and $w_2(t)$ represents the vector of quantum noise inputs for $G_2$. Now,
since $G_1$ and $G_2$ each only interact with the controller, it follows that the dynamics of the
closed-loop system can be written in the form: 
\begin{eqnarray}
 dx(t)&=& \tilde Ax(t) dt + \tilde B dw(t); x(0)=x,
\label{eq:closed-loop}
 \end{eqnarray}
where the real matrices $\tilde A$ and $\tilde B$ have the special structure:
\begin{eqnarray*}
\tilde A &=& \left[ \small \begin{array}{ccc} A_{11} & 0 & A_{13} \\
0 & A_{22} & A_{23}\\
A_{31} & A_{32} & A_{33} 
\end{array} \normalsize \right];\quad 
\tilde B =  \left[\small \begin{array}{cc} B_{11} & 0  \\ 0 & B_{22} \\ B_{31} & B_{32}   \end{array} \normalsize \right],
\end{eqnarray*}
with  $A_{11}=A_1,\, A_{22}=A_2\,,B_{11}=B_1;\, B_{22}=B_2$.
Our main result in
this section is the following.

\begin{theorem}
\label{T3}
Consider any classical  LTI controller that is connected to
the linear quantum system $G$ such that $\tilde A$ is Hurwitz in the
closed loop system  (\ref{eq:closed-loop}). 
Then the resulting closed-loop dynamical bipartite Gaussian quantum system
is separable at steady state. Thus, a classical LTI controller cannot generate an entangled steady state  from any initial  Gaussian
state. 
\end{theorem}
\begin{proof}
Since the controller state $z(t)$ is classical, the commutation matrix $\Theta$ for $x(t)$ will be degenerate canonical \cite[Sections II, III, and III C]{JNP06} of the form 
$
\Theta=\left[\small \begin{array}{ccc} \Theta_1 & 0 & 0 \\ 0 & \Theta_2 & 0 \\ 0_{n_c \times 2} & 0_{n_c \times 2} & 0_{n_c \times n_c} \end{array} \normalsize \right], 
$
 where  $\Theta_1=\Theta_2=J$. Suppose that the controller state
 $z(t)$ is of arbitrary dimension $n_c$.  We  have that the closed-loop mixed quantum-classical system satisfies the constraint \cite[Theorem 3.4]{JNP06} : 
 $
 \tilde A\Theta + \Theta \tilde A^T -\isym \tilde B T_w \tilde B^T =0,
 $
with $T_w =\diag(T_{w_1},T_{w_2})$.
This equation is equivalent to the following:
\begin{eqnarray}
&&\left[\small
\begin{array}{cc}
A_{11} \Theta_1+\Theta_1 A_{11}^T -\isym B_{11}T_{w_1}B_{11}^T  \\
0 \\
A_{31}\Theta_1 -\isym B_{31} T_{w_1} B_{11}^T \end{array}\normalsize
\right.\notag \\
&&\quad \small \begin{array}{c} 0 \\  
A_{22} \Theta_2 + \Theta_2 A_{22}^T -\isym B_{22} T_{w_2} B_{22}^T \\ 
A_{32} \Theta_2 -\isym B_{32}T_{w_2}B_{22}^T \end{array} \normalsize \notag\\
&&\quad \left. \small \begin{array}{c}  \Theta_1 A_{31}^T -\isym   B_{11} T_{w_1} B_{31}^T \\  
\Theta_2 A_{32}^T-\isym B_{22} T_{w_2} B_{32}^T \\
-\isym B_{32} T_{w_2} B_{32}^T   
\end{array}\normalsize\right]=0. \notag \\
\label{eq:closed-loop-cons1}
\end{eqnarray}
Multiplying the $(2,2)$, $(2,3)$, $(3,2)$, and $(3,3)$ elements of this matrix equation by $-1$, yields
\begin{eqnarray}
&&\left[\small
\begin{array}{c}
A_{11} \Theta_1+\Theta_1 A_{11}^T -  \isym B_{11}T_{w_1}B_{11}^T  \\
0 \\
A_{31}\Theta_1-\isym B_{31} T_{w_1} B_{11}^T  
\end{array} \normalsize
\right. \notag\\
&&\quad \small \begin{array}{cc}  0\\
-A_{22} \Theta_2 -\Theta_2 A_{22}^T+\isym B_{22} T_{w_2} B_{22}^T \\
-A_{32} \Theta_2 +\isym B_{32}T_{w_2}B_{22}^T  
\end{array} \normalsize  \notag \\
&&\quad \left. \small   \begin{array}{c} \Theta_1 A_{31}^T -\isym B_{11} T_{w_1} B_{31}^T \\
-\Theta_2 A_{32}^T +\isym B_{22} T_{w_2} B_{32}^T\\
\isym B_{32} T_{w_2} B_{32}^T
\end{array} \normalsize \right]=0.
\label{eq:closed-loop-cons2}
\end{eqnarray}
Letting $\hat \Theta= \diag(\Theta_1,-\Theta_2,0_{n_c \times n_c})$ 
 and $\hat T_w= \diag(T_{w_1},-T_{w_2})$, this matrix equality can be written as:
$\tilde A\hat \Theta + \hat \Theta \tilde A^T -\isym \tilde B \hat T_w \tilde B^T =0.$  
We now use the fact  that the closed-loop matrix $\tilde A$ is Hurwitz. Then, as discussed in Section \ref{sec:proof-HUP}, the  symmetrized steady state covariance matrix $P$ satisfies:
$\tilde AP+P \tilde A^T + \tilde B S_{w} \tilde B^T =0.$
Defining $\tilde P=P + \isym \hat \Theta$ and $\tilde F_w = S_w + \hat T_w$, we have that $\tilde P$ satisfies the complex Lyapunov equation:
$\tilde A\tilde P + \tilde P \tilde A^T + \tilde B \tilde F_w \tilde B^T =0.$ 
Recalling that  $S_w=\diag(S_{w_1},S_{w_2})$ and $F_w \geq 0$, we note that $F_{w_i}=S_{w_i}+T_{w_i} \geq 0$ for $i=1,2$.
Then, we note that since $\tilde F_w=\diag(F_{w_1},F_{w_2}^{\#})$, we have that $\tilde F_w \geq 0$; see, e.g., \cite{Bern05}. Therefore, since  $\tilde A$ is Hurwitz, we have that $\tilde P \geq 0$.
Partitioning $P$ according to the partitioning of $x(t)$ into its quantum and classical components as 
$P=\left[ \small\begin{array}{cc} 
 P_{11} &  P_{12} \\  P_{21 } &   P_{22}  \end{array} \normalsize \right]$, the property $\tilde P = P + \isym \hat \Theta \geq 0$ implies that 
$
  P_{11} +\isym \diag(\Theta_1,-\Theta_2) 
=  P_{11} + \isym \diag(J,-J)  \geq 0.
$
 Therefore it follows from Lemma \ref{lem:sep-covar} that the
dynamical bipartite  Gaussian quantum  system is separable at
steady-state. Thus, a classical LTI controller cannot generate an entangled steady state  from any initial  Gaussian
state. \end{proof}
 
Removing the classical controller by defining all its system matrices
to be zero, a special case of Theorem \ref{T3} shows  that two independent
and unconnected  systems $G_1$ and $G_2$, with $A_1$ and $A_2$ Hurwitz, which
are initially entangled become separable in the steady state. 
 
\section{Classical finite dimensional linear controllers cannot generate any entanglement in bipartite Gaussian quantum systems in finite time}
\label{sec:LTI-nonentg-finite}
In the previous section, we have shown that starting from any state,
separable or entangled, a classical LTI controller cannot generate or maintain
entanglement  at
steady state in a linear  dynamical bipartite  Gaussian quantum system. In this section, by a  slight modification of the arguments of
the previous section, we will show that classical  finite dimensional linear controllers  cannot
generate bipartite entanglement in a finite time for any initially
separable linear  dynamical bipartite  Gaussian quantum
system. Moreover, for this finite time analysis, we may drop the
requirement that the controller is chosen so that the closed-loop  matrix $\tilde A$ is Hurwitz.  

We follow the notation and set up of the last section. Instead of considering the
steady state covariance matrix $P$, we now consider the symmetrized finite time covariance matrix $P(t) =\frac{1}{2} \tr \bigl(\rho(0) ( x(t) x(t)^T + (x(t)x(t)^T)^T ) \bigr)$, $0 \leq t < \infty$ satisfying the Lyapunov differential equation:
\begin{eqnarray*}
\dot{P}(t) &=& \tilde A P(t) + P(t) \tilde A^T + \tilde B S_w \tilde B^T,\quad P(0)=P_0.
\end{eqnarray*} 

\begin{theorem}
Suppose that a linear dynamical bipartite Gaussian quantum system is initially
separable. Then it remains separable for all $t \geq 0$ under the action of any classical  LTI controller.
\end{theorem}
\begin{proof}
Since the system is initially separable, $P(0)+\isym \hat \Theta \geq 0$ by Lemma \ref{lem:sep-covar}. Let $\tilde P(t)=P(t) + \isym \hat \Theta$. Then following the same
lines of argument as in the proof of Theorem \ref{T3}, and by using standard results on Lyapunov differential equations, we find that since $\tilde P(0) \geq 0$ and $\tilde F_{w} \geq 0$ (hence also $\tilde B\tilde F_{w} \tilde B^T \geq 0$)  that
$
\tilde P(t) \geq 0,
$
for all $t \geq 0$ regardless of the values of $\tilde A$ and $\tilde B$. Similarly partitioning $P(t)$ according to the partitioning of $x(t)$ into its quantum and classical components as  $P(t)=\left[ \small \begin{array}{cc}  P_{11}(t) &  P_{12}(t) \\  P_{12}(t) &  P_{22}(t) \end{array} \normalsize \right]$, it follows that
$
 P_{11}(t)+\isym \diag(J,-J) \geq 0,
$
for all $t \geq 0$. This shows that when the bipartite Gaussian system is initially in a
separable state, then under the action of a classical  LTI
controller it will remain so for all times. 
\end{proof}
\begin{remark}
Note that it is straightforward to extend the proof of the above
theorem to allow for linear time-varying controllers rather than LTI
controllers. 
\end{remark}
\begin{example}
\label{eg:two-cavities}
Consider the quantum optical system shown in  Fig.~\ref{fig:two-cavities}. Suppose that the two optical cavities are identical with the partially reflecting mirror on each cavity having coupling coefficient  $0.01$. Also, let the position and momentum operators of cavity $G_i$ be $(q_i,p_i)$, and let $\tilde x=(q_2,p_2,q_1,p_1)^T$. Let $w_{i1}(t)=2\Re\{\mathcal{A}_{i}(t)\}$ and $w_{i2}(t)=2\Im\{\mathcal{A}_{i}(t)\}$, $i=1,2,3$. Then the dynamics of the two degree of freedom linear quantum stochastic system (without the controller, homodyne detector and modulator attached) is given by:
\begin{align*}
d\tilde{x}(t) &=- 0.005 \tilde{x}(t)dt  - 0.1 d(w_{21}(t),w_{22}(t),w_{11}(t),w_{12}(t))^T, \\
dy^o(t) &= 0.1 \tilde{x}(t) dt + d (w_{21}(t),w_{22}(t),w_{11}(t) ,w_{12}(t))^T.
\end{align*}  
Let $y^o(t)=(y^{o}_1(t),y^{o}_2(t),y^{o}_3(t),y^{o}_4(t))^T$. The amplitude quadrature $y^o_3(t)$ of $y^o(t)$ is measured using the homodyne detector and is used as the (stochastic) input $m_1(t)=y^{o}_3(t)$ to a first order LTI controller that produces a two dimensional output signal $u(t)=(u_1(t),u_2(t))^T$. The dynamics of the controller is:
\begin{align*}
dz(t)&= Az(t) dt + Bdm_1(t),\;z(0)=0, \\
u(t) &= \left[\begin{array}{cc} C_1^T & C_2^T \end{array}\right]^Tz(t),
\end{align*}
where $z(t)$ denotes the state of the controller, and $A=-1$,
$B=C_1=C_2=1$. The output signal $u(t)$ is passed through an
electro-optic modulator and sent to the partially reflecting mirror of
cavity $G_2$. Let $x(t)=(\tilde x(t),z(t))^T$. We then have that  the
interconnection of the controller with the two cavities via the
homodyne detector and electro-optic  modulator is  a mixed
quantum-classical system with dynamics of the form
(\ref{eq:closed-loop}) defined by the matrices 
\begin{align*}
\tilde A &= \left[ \small \begin{array}{ccccc} -0.005  & 0 & 0 & 0 & -0.1C_1\\ 0 & -0.005 & 0 & 0 & -0.1C_2 \\ 0 & 0 & -0.005 & 0 & 0 \\ 0 &0 & 0&  -0.005 & 0 \\
0  & 0 & 0.1 B  &  0 & A  \end{array} \normalsize \right]; \\
\tilde B &= \left[ \small \begin{array}{cccc}   -0.1 & 0 & 0 & 0 \\ 0 &  -0.1 & 0 & 0  \\   0 & 0 & -0.1 & 0  \\  0 & 0 &  0 & -0.1   \\ 0 & 0 & B & 0 \end{array} \normalsize \right],
\end{align*}
and is driven by the noise $(w_{31}(t),w_{32}(t),w_{11}(t),w_{12}(t))^T$. Suppose that the bipartite state of the two cavities is  in an {\em initially entangled} bipartite Gaussian state with covariance matrix $P_{11}(0)$ given below in (\ref{eq:init-entangled})
\begin{equation}
P_{11}(0)=\left[\small \begin{array}{cccc}  0.5028 &  0 & -0.0528 & 0 \\ 0  &  0.5028  &  0  &  0.0528 \\
   -0.0528 &  0 &  0.5028 & 0 \\ 0  &   0.0528  &  0 &  0.5028 \end{array} \normalsize \right]. \label{eq:init-entangled}
\end{equation}
We take as our measure of entanglement the logarithmic negativity $E_{N}$ \cite{VW02,PV07,YNJP08}. Partitioning $P_{11}(t)$ into $2 \times 2$ blocks as $\left[ \small \begin{array}{cc} P_{11,1}(t) & P_{11,2}(t) \\ P_{11,2}(t)^T & P_{11,3}(t) \normalsize \end{array}\right]$, $E_{N}(P_{11}(t))$ is given by $E_{N}(P_{11}(t))=\max(0,-\ln (2\nu(t)))$, where $\nu(t) = \frac{1}{\sqrt{2}}\sqrt{\tilde \Delta(t) - \sqrt{\tilde \Delta(t)^2-4\det(P_{11}(t))}}$ and $\tilde \Delta(t)= \det(P_{11,1}(t))+\det(P_{11,3}(t))-2\det(P_{11,2}(t))$. Note that the logarithmic negativity is always nonnegative and has a  value of zero if and only if the state is separable \cite{VW02,PV07}, otherwise the state is entangled, with a higher value of $E_N$ indicating a higher degree of entanglement. The initial value of the logarithmic negativity is  $E_{N}(P_{11}(0))=0.1054$. The solid line in Fig.~\ref{fig:example} shows that under the action of this classical controller, the logarithmic negativity steadily decreases and finally goes to zero in a finite time. At this point, the state becomes separable and  remains so for all future times. If we instead start  at an {\em initially separable}  state with covariance matrix $P_{11}(0)$ as given in (\ref{eq:init-separable})
\begin{equation}
P_{11}(0)=\left[\small \begin{array}{cccc}  
0.5704   &      0  &   0.0034   &  0.0562 \\
         0   &  0.5704    &     0  &  0.0528 \\
    0.0034    &     0  &   0.6203  &  0.0499 \\
    0.0562  &  0.0528  &  0.0499  &  0.6203
 \end{array} \normalsize \right], \label{eq:init-separable}
\end{equation}
then the oscillators' joint state remains separable, as shown in the  dashed  line in Fig.~\ref{fig:example}. 
\begin{figure*}[htbp]
\centering
\includegraphics[scale=0.5]{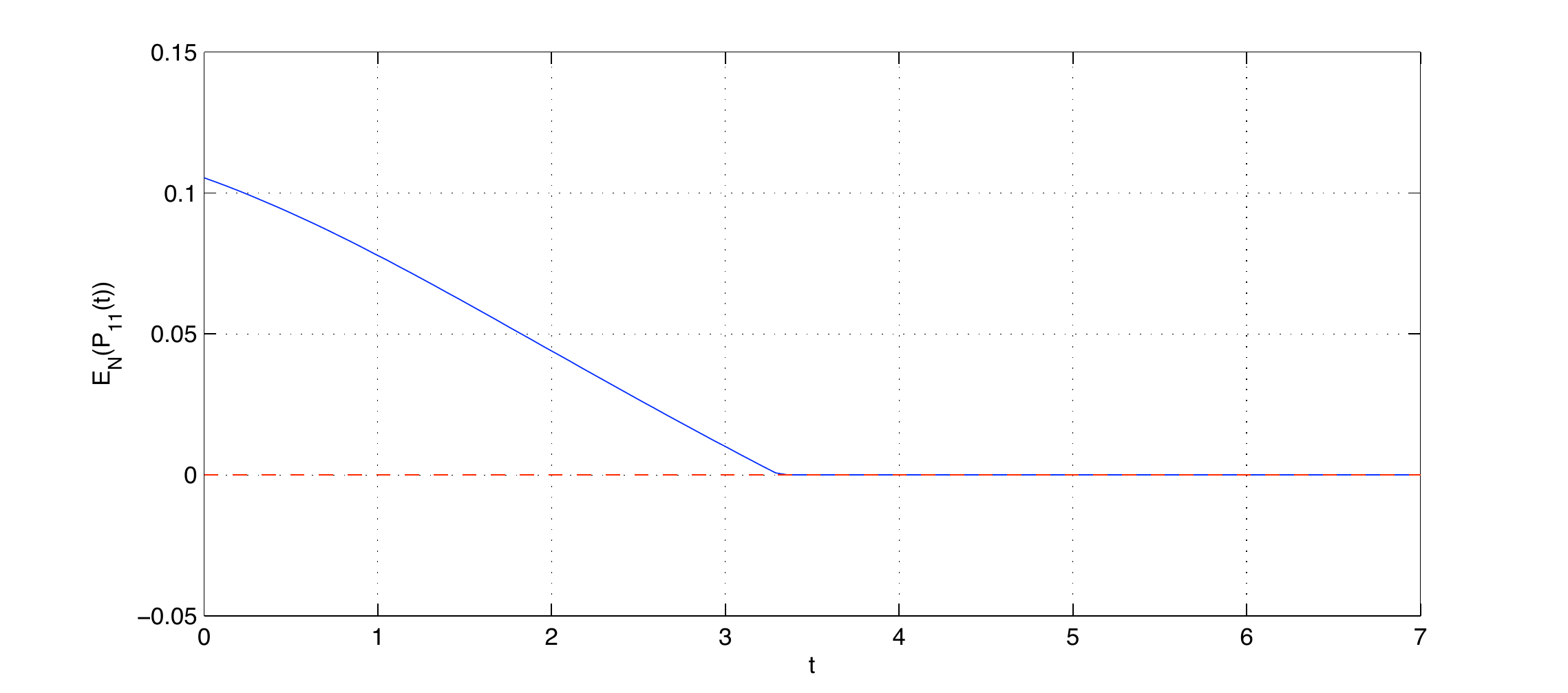}
\caption{The evolution of  logarithmic negativity $E_N(P_{11}(t))$ over time for the system considered in Example \ref{eg:two-cavities}. The solid  line shows the evolution starting from an entangled state with covariance matrix $P_{11}(0)$ in (\ref{eq:init-entangled}).  The  dashed line shows the evolution starting from a separable state with covariance matrix $P_{11}(0)$ in (\ref{eq:init-separable}).} \label{fig:example}
\end{figure*}
\end{example}

\section{Conclusions}
\label{sec:conclusions}

By employing system-theoretic arguments and methods,  we
were  able to 
give a systems theory proof of the fact that classical  LTI controllers cannot generate
steady state entanglement in linear  dynamical bipartite  Gaussian quantum systems. Furthermore,
we also give a systems theory proof of the fact that classical  linear controllers cannot generate
 entanglement in a dynamical bipartite  Gaussian system initially  in a
separable state. An interesting topic for future research is to consider system-theoretic analysis of entanglement between the continuous-mode output
fields.

\bibliographystyle{IEEEtran}
\bibliography{ieeeabrv,rip,mjbib2004}

\end{document}